\documentclass[12pt,a4paper]{article}
\usepackage{authblk}

\usepackage{graphicx}
\usepackage{hyperref}
\usepackage{amssymb,amsmath,amstext,amsfonts}

\usepackage{a4wide}

\usepackage{fullpage}
\usepackage{amsfonts,amsmath,amssymb,amsthm,boxedminipage,color,url,fullpage}
\usepackage{enumerate,paralist}
\usepackage[numbers]{natbib} 
\usepackage[labelfont=bf]{caption}
\usepackage{aliascnt,cleveref}

\usepackage{breqn}

\usepackage{epic,eepic}

\newtheorem{lemma}{Lemma}[section]

\newtheorem{theorem}[lemma]{Theorem}

\theoremstyle{definition}

\theoremstyle{remark}

\hypersetup{colorlinks=true,linkcolor=blue,citecolor=blue,filecolor=blue,urlcolor=blue}

\begin{document}

\title{\bf Tighter Bounds on the Inefficiency Ratio of Stable Equilibria in Load Balancing Games}

\author{Akaki Mamageishvili}
\author{Paolo Penna}

\affil{Department of Computer Science, ETH Zurich, Switzerland}

\date{November, 2015}

\maketitle

\begin{abstract}
In this paper we study the inefficiency ratio of stable equilibria in load balancing games introduced by  \citet{original}. We prove tighter lower and upper bounds of $7/6$ and $4/3$, respectively. This improves over the best known bounds in problem ($19/18$ and $3/2$, respectively). Equivalently, the results apply to the question of how well the optimum for the $L_2$-norm can approximate the $L_{\infty}$-norm (makespan) in  identical machines scheduling. 
\end{abstract}

%
\section{Introduction}

Load balancing problems are actively studied in the context of games. Transition from a centralized optimization problem to a decentralized environment is immediate, by assuming that the jobs are owned by selfish but rational players. These games are prototypical of resource allocation problems in which users (players) do not act altruistically, therefore leading the system to  suboptimal configurations. 
Naturally, one can consider the \emph{social optimum} as the allocation minimizing the \emph{makespan}, that is, the maximum load over all machines (a classical measure of efficiency). In contrast, in the game-theoretic setting, the players will optimize for their own job only, typically reaching what is called a \emph{Nash equilibrium}, that is, an allocation such that no player has an interest in moving to another machine.\footnote{In this work we consider  \emph{pure} Nash equilibria, that is, configurations in which each player chooses one strategy and unilateral deviations are not beneficial. In general, games may also posses mixed Nash equilibria in which players choose strategies according to a probability distribution.}  
The study of how inefficient these Nash equilibria can be is a central problem in algorithmic game theory, as it measures how much selfishness can impede optimization.

\citet{original} introduced and studied the
\emph{inefficiency ratio of stable equilibria (IRSE)} in several games, including load balancing ones (see Section~\ref{sec:preli} for formal definitions). This notion quantifies the efficiency loss in games when players play certain \emph{noisy} best-response dynamics (see Section~\ref{sec:related} for a detailed discussion). For load balancing games, 
the IRSE has another very simple and natural interpretation (which is also of independent interest and studied earlier):

\begin{quote}
	\emph{Are the allocations minimizing the $L_2$-norm (sum of the squares of the machine loads) also sufficiently good for minimizing the $L_\infty$-norm (makespan)?}
\end{quote}
Intuitively, the IRSE on load balancing games is equal to some value $c$ if  \emph{every} allocation minimizing the $L_2$-norm is automatically a $c$-approximation for the $L_\infty$-norm (i.e., the makespan of this allocation is at most $c$-times the optimal makespan). An exact bound on the IRSE is not known (as opposed to othe measures related to the inefficiency of equilibria -- see next section).  \citet{original} proved an upper bound of $\frac{3}{2}$ on IRSE and observed that an example in \cite{AloAzaWoeYad97} implies a lower bound on IRSE of $\frac{19}{18}$.
In this work we improve both bounds: In Section~\ref{sec:lower-bound} we show an improved lower bound of $7/6$, while the upper bound of $4/3$ is proved in Section~\ref{sec:upper-bound}.

\subsection{Significance of the results and related work}\label{sec:related}
 The inefficiency of Nash equilibria is often measured through two classical notions: The \emph{pice of anarchy (PoA)} introduced in \citet{PoAoriginal} for load balancing games on related machines, and the \emph{price of stability (PoS)} introduced in \citet{PoSoriginal}, which compare respectively the \emph{worst} and the \emph{best} Nash equilibrium to the social optimum. In some cases, these notions can be considered too extreme as they may include "unrealistic" equilibria. 
 
  The IRSE \cite{original} studies the quality of equilibria selected by certain \emph{noisy} best-response dynamics \cite{Blu98}. Intuitively, these dynamics will most likely rest on pure Nash equilibria  \emph{minimizing the potential} of the game, and the IRSE can be seen as the price of anarchy restricted to these selected equilibria (in fact, IRSE is also known in the literature under the name of \emph{potential optimal PoA} by \citet{KawMak13}, who also considered the analogous \emph{potential optimal PoS}). In load balancing games (and several others) we have $PoS < IRSE < PoA$, which in a sense tells that the PoA and the PoS are either too pessimistic or too optimistic. Specifically, on $m$ machines   $PoA=\left(2-\frac{1}{m+1}\right)$ \cite{PoA,agtBook}, $PoS=1$ \cite{agtBook}, while $IRSE$ is between $19/18$ and $3/2$ \cite{original}.
 The latter bounds are strengthened in the present paper to $7/6\leq IRSE\leq 4/3$. This means that players can easily compute a $4/3$ approximation of the optimum via these simple dynamics, but also that optimal (exact) solutions are unlikely to be chosen either, in some instances. 
 
The upper bound $4/3$ also suggests and intriguing comparison with the study of \emph{sequential PoA} by 
\citet{Hassin} for these games: There players are far from myopic and the equilibrium is meant on an extensive form games in which players are able to reason about future moves of following players. Our upper bound says that myopic best-response are not worse than in their dynamics where players are capable of more sophisticated reasoning.



\section{Preliminaries}\label{sec:preli}
In load balancing we have $n$ jobs of weight $w_1,\ldots,w_n$ that we want to put on $m$ identical machines (each job is allocated to one machine). The job allocation determines the load $l_j$ of each machine $j$, that is the sum of the jobs weights that are allocated to this machine. The goal is to find an allocation that has the lowest possible \emph{makespan}, that is, the maximum load over all machines.

In \emph{load balancing games} each job is a \emph{player} who  can choose any of the $m$ possible machines. The cost for player $i$ is simply the load of the machine chosen by this player, and naturally each player aims at minimizing \emph{her own cost}. The strategies of all players specify a job allocation (in the game theoretic terminology this is the strategy profile).

An allocation minimizing the makespan is called a \emph{social optimum}, and its makespan is called a \emph{social optimum makespan}. The \emph{potential} associated to an allocation is the sum of the \emph{squares} of the corresponding machine loads,  ${l_1}^2+\cdots+{l_m}^2$  where $l_j$ is the load of machine $j$ at this allocation. 
An allocation minimizing the potential function is called a \emph{potential optimum}.

It is well known that load balancing games are weighted potential games with the above  potential function. This means that all pure Nash equilibria, allocations where no player can unilaterally improve moving to another machine, are actually `local optimal' for the potential (a single job move cannot reduce the potential).   
Asadpour and Saberi \cite{original} introduced and studied the \emph{inefficiency ratio of stable equilibria (IRSE)}, which is the ratio between the \emph{worst} makespan of a potential optimal allocation  divided by the social optimum makespan. 

Potential optimum allocations  satisfy the following condition (see  Figure~\ref{fig:bundle-swap}). Split the total load of each machine into two bundles of jobs, that is, $l_j=x_j+y_j$ where $x_j$ is the sum of the weights of a (possibly empty) subset of jobs allocated to $j$. If two machines $i$ and $j$ satisfy $x_i < x_j$, then $y_i>y_j$ for otherwise swapping $x_i$ with $x_j$ reduces the potential. Pure Nash equilibria satisfy the weaker condition that a single job in machine $i$ does not improve if moving to another machine $j$, that is, $l_i -w_i\leq l_j$.   

\begin{figure}
	\centering
	\includegraphics[scale=.5]{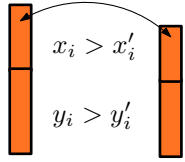}
	\caption{When a swap of bundles of jobs reduces the potential.}
	\label{fig:bundle-swap}
\end{figure}


\section{Improved lower bound}\label{sec:lower-bound}

In this section we strengthen the $19/18$ lower bound on IRSE in \cite{original,AloAzaWoeYad97}.

\begin{figure}
    \centering

    \includegraphics[scale=.5]{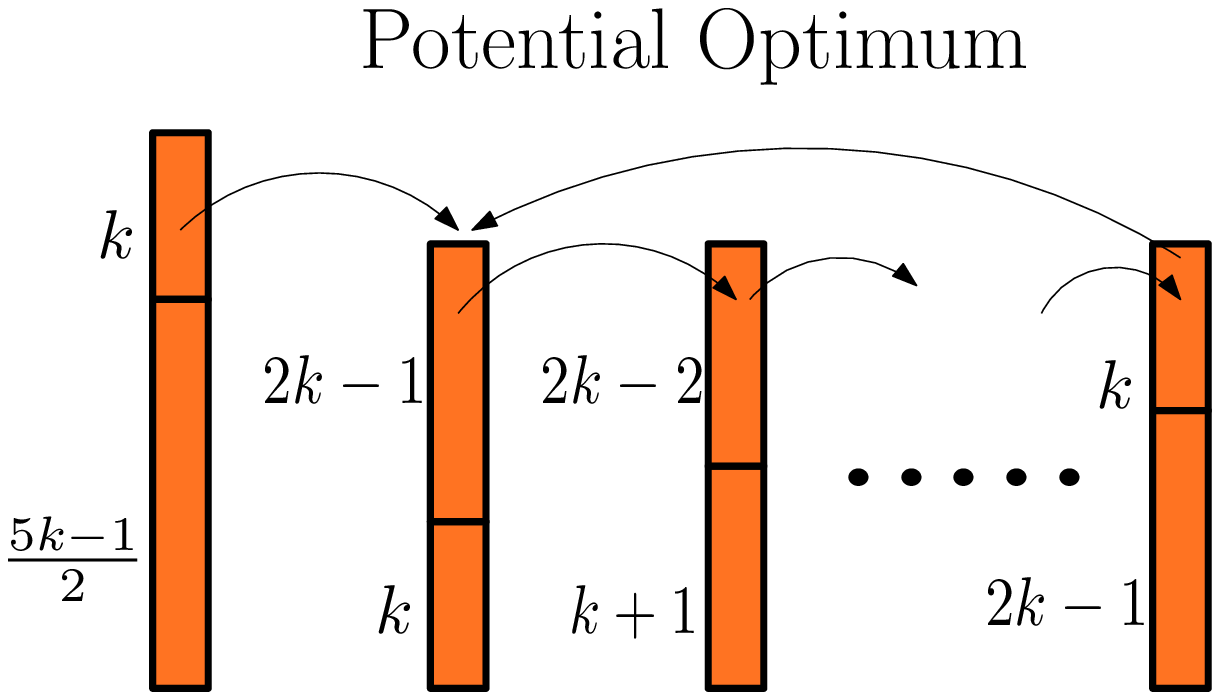}
    \hspace{0.1\textwidth}
    \includegraphics[scale=.5]{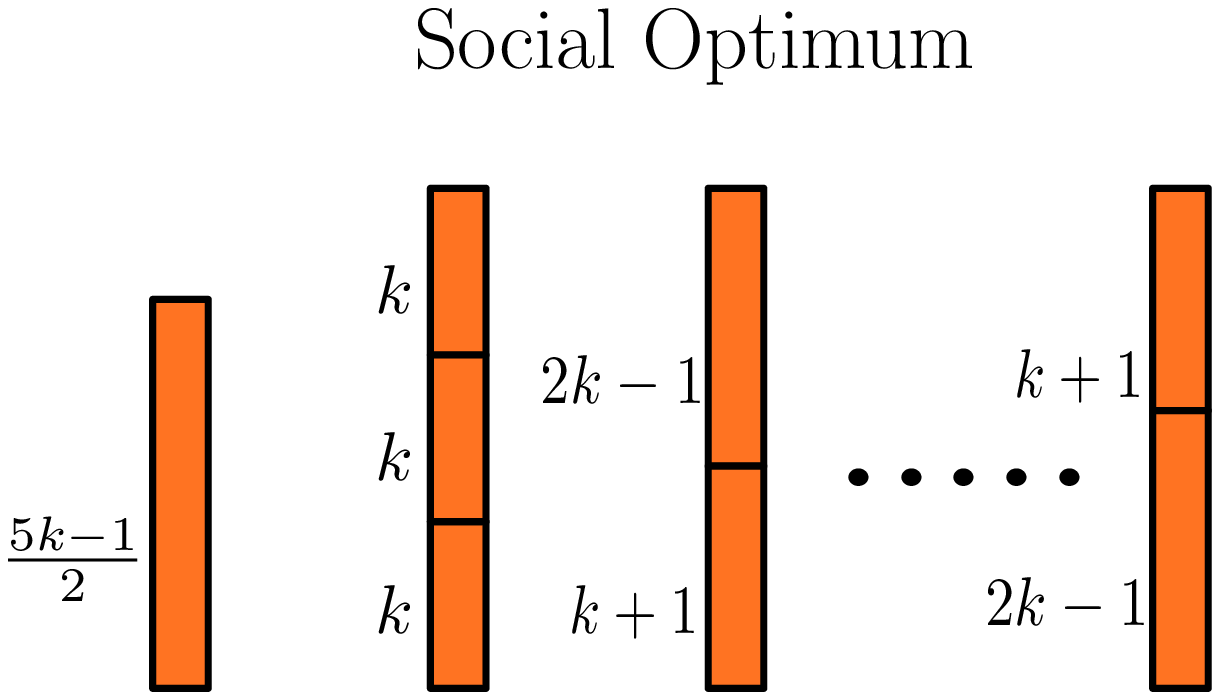}
    \caption{Lower bound 7/6 on IRSE.}
    \label{fig:lower_bound}
\end{figure}

\begin{theorem}
IRSE in load balancing games is at least $7/6$.
\end{theorem}

\begin{proof}
Consider the instance in Figure~\ref{fig:lower_bound} where  $m=k+1\geq 3$,  $n=2k+2$ and the weights of jobs are $k,k,k,k+1,k+1,\ldots,2k-1,2k-1,(5k-1)/2$. We prove that the left allocation of jobs on the figure minimizes the potential function, while the one on the right has optimal makespan,  thus implying 
\[
IRSE \geq \frac{(7k-1)/2}{3k} = \frac{7m -8}{6m-6}
\]
which tends to $7/6$ as $m$ goes to infinity.
 First note that the potential in both allocations are equal. Consider any job allocation and without loss of generality assume that the job with the largest  weight is on machine $1$. If we fix a load on machine $1$, then any job allocation which balances the load on the other machines minimizes the potential (among all allocations with this fixed load on machine $1$). Therefore, if the largest job is alone on machine $1$ then the potential is minimized in the social optimum case, while if it is located on  machine one together with job of smallest weight then the potential is optimized when the job allocation is like on the left side of Figure~\ref{fig:lower_bound}. If the load on machine $1$ is strictly larger than $5(k-1)/2+k$ then the potential is not a potential optimal allocation. We conclude that both allocations are potential optimal, in particular the one on the left which gives the lower bound on IRSE.
\end{proof}

\section{Improved upper bound}\label{sec:upper-bound}

Without loss of generality assume that a makespan in social optimum is equal to $1$ and loads of machines in the potential optimizer job allocation are sorted in a non-increasing order $l_1\geq l_2\geq \cdots \geq l_m$.

\begin{theorem}
IRSE in load balancing games is at most $4/3$.
\end{theorem}

\begin{proof}
We have to prove that $l_1\leq 4/3$. Assume that it does not hold, that is $l_1=1+\alpha>4/3$ implying $\alpha>1/3$. Machine $1$ contains at least two jobs, otherwise the optimum makespan would not be $1$. The smallest job on machine $1$, which we call $s$, is strictly larger than $\alpha$, otherwise since $l_m<1$ we can decrease the potential by moving $s$ to the last machine. Assume that the weight of $s$ is $\alpha+\beta$ where $\beta>0$ and since $\alpha+\beta \leq l_1-\alpha-\beta=1-\beta$ it implies $\beta < \alpha$. The proof of the theorem is based on the following property of potential optimal allocations (whose proof is postponed). 

\begin{figure}
	\centering
	\includegraphics[scale=.6]{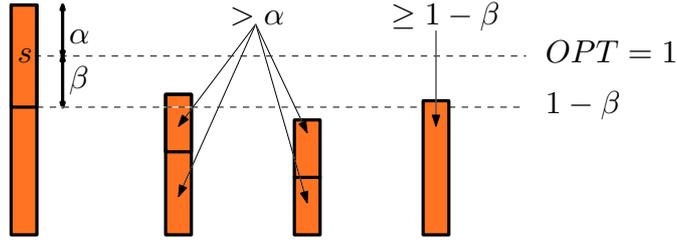}
	\caption{Proof of upper bound --  Lemma~\ref{le:high-jobs}.}
\end{figure}

\begin{lemma}\label{le:high-jobs}
For any $\alpha > 1/3$, in a potential optimal allocation, every machine  must have one job of weight at least $1-\beta$ or two jobs of weight strictly bigger than $\alpha$.
\end{lemma}
Note that for $\alpha >1/3$ this contradicts the fact that the optimum makespan is equal $1$ since any job from machine $1$ has weight larger than $\alpha$ and together with these jobs cannot be allocated on $m-1$ machines without exceeding $1$ (clearly the jobs in machine $1$ cannot be together in the optimum makespan, no three jobs of weight $\alpha$ can fit together,  nor a job of $1-\beta$ with a job of $\alpha$). 
\end{proof}

\subsection{Proof of Lemma~\ref{le:high-jobs}}
We next prove Lemma~\ref{le:high-jobs}. First observe that $l_i \geq 1-\beta$ for otherwise we can move job $s$ from machine $1$ to machine $i$ and decrease the potential. Sort the jobs on machine $i$ in increasing weight, and consider a bundle of small jobs such that the load is still above or equal to $1-\beta$ after removing them from the machine (this bundle can be empty). Let $x$ be the overall weight of this bundle, and $y$ be the size of the next job whose removal reduces the load below $1-\beta$. That is 
\begin{align}\label{eq:proof-claim-UB:bundle}
	l_i - x \geq& 1-\beta & \text{ and }& & l_i - x - y <& 1 - \beta.
\end{align} 
Also note that 
\begin{equation}\label{eq:proof-claim-UB-small-vs-bundle}
	x + y \geq \alpha + \beta
\end{equation}
otherwise we can exchange job $s$ on machine $1$ with job bundle of weight $x + y$ on machine $i$ and strictly decrease the potential. If $y$ was the only job after removing $x$, then by definition $y \geq 1-\beta$ and the lemma holds. Otherwise, we show that $y \geq \alpha$. 
We next distinguish two cases depending on the load $l_i$ of machine $i$ (see Figure~\ref{fig:UB:lemma-proof}):

\begin{figure}
	\centering
	\includegraphics[scale=.6]{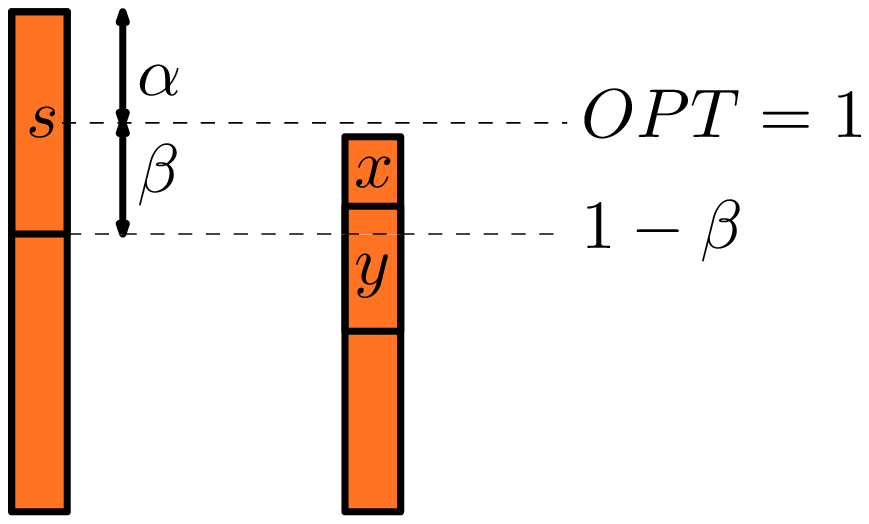} \hspace{1.5cm}
	\includegraphics[scale=.6]{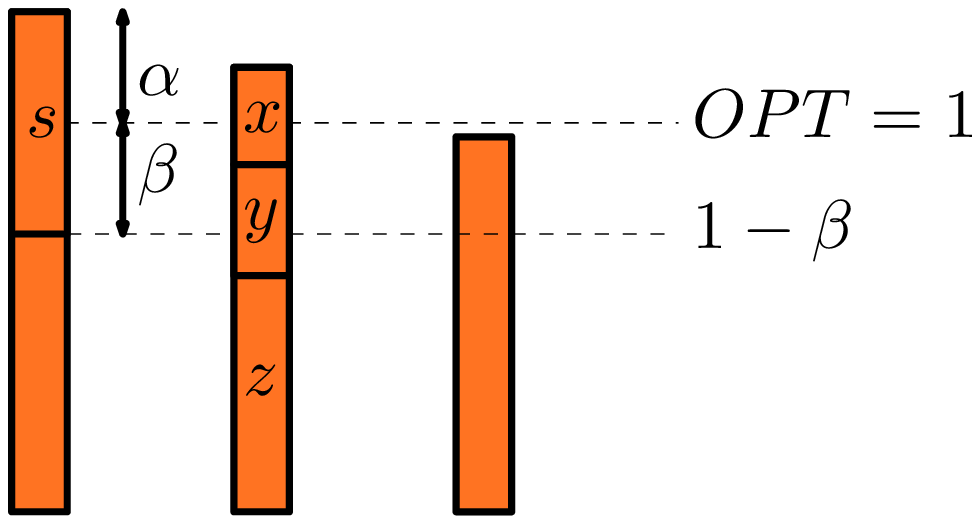}
	\caption{The two cases to prove Lemma~\ref{le:high-jobs}.}
	\label{fig:UB:lemma-proof}
\end{figure}

\begin{description}
	\item[($l_i < 1$).] From the first inequality in \eqref{eq:proof-claim-UB:bundle} we get $x < \beta$, and then by \eqref{eq:proof-claim-UB-small-vs-bundle} $y>\alpha$. 

    \item[($l_i\geq 1$).]  
    If $x<\beta$ then $y>\alpha >1/3$ using the same argument of the previous case. Assume therefore $x\geq \beta$ and $y \leq 1/3 < \alpha$. Let $z$ be the load of machine $i$ after removing the bundle of size $x+y$, that is, $l_i= x + y + z$.  We rearrange the jobs in the potential optimum and show that this will further reduce the potential, thus a contradiction. The job rearrangement goes as follows: 
\begin{enumerate}
	\item Job $s$ moves from machine $1$ to machine $i$;
	\item The smallest between $x$ and $y$ moves to machine $m$, while the biggest moves to machine $1$. 
\end{enumerate}
Consider first the case $x\leq y$. Then the difference between potential of the original allocation and potential of the new one is 
\begin{align} 
	(1+\alpha)^2+(x+y+z)^2+l_m^2-(1-\beta+y)^2-(z+\alpha+\beta)^2-(l_m+x)^2 & = \nonumber
	\\ 2\Big(\alpha+xy+yz+xz+\beta+\beta y- (\beta^2+y+\alpha z+\alpha\beta+z\beta+l_mx)\Big) & > \nonumber
	\\ 2\Big(\alpha+xy+yz+xz+\beta+\beta y- (\beta^2+y+\alpha z+\alpha\beta+z\beta+x)\Big)
	\label{potential}
\end{align}where the inequality is due to $l_m <1$. 
Next we prove that this quantity is positive and obtain the desired contradiction. Note that $y+z\leq l_m$ otherwise moving the  bundle of weight $x$ to machine $m$ reduces the potential. Therefore \eqref{potential} is linear in $x$ with slope $y+z - 1 \leq 0$ and since $x \leq y$, we can bound it from below by
\begin{align}
	2\Big(\alpha+y^2+2yz+\beta+\beta y- (\beta^2+y+\alpha z+\alpha\beta+z\beta+y)\Big) & \geq \nonumber
	\intertext{since this quantity is linear in $z$ with slope $2y - (\alpha + \beta)\geq x + y - (\alpha + \beta)\geq 0$ and $z \geq 1- \beta - x \geq 1 - \beta - \alpha$} 
	2\Big(\alpha+2y - \beta y - y^2+\beta- (\beta^2+y+\alpha  - \alpha y+\beta- \beta^2 - \beta y+y)\Big) & = \nonumber \\
	2\Big(y - \beta y - y^2+\beta- (- \alpha y+\beta- \beta y+y)\Big)& = \nonumber \\
	2\Big(\alpha y  - y^2 \Big)&  \nonumber
\end{align}
which is strictly positive since $y< \alpha$.

		The case $y <x$ is similar to the previous one. Specifically, now the potential difference is strictly bigger than \eqref{potential} but with $x$ and $y$ exchanged,
		\begin{align} 
			2\Big(\alpha+xy+yz+xz+\beta+\beta x- (\beta^2+x+\alpha z+\alpha\beta+z\beta+y)\Big) &  \label{potential:y<x}
		\end{align}
		and this quantity is linear in $y$ with non-positive slope  $x+z-1\leq 0$, since otherwise moving bundle $y$ to machine $m$ would reduce the potential. For the case $x < \alpha$ the rest of the proof is identical to the previous case with $x$ and $y$ exchanged. For $x \geq \alpha$ we use $y < \alpha$ and  obtain a lower bound for \eqref{potential:y<x} by replacing $y$ with $\alpha$, and then $z$ by $1 - \beta - \alpha$:
		\begin{align*} 
			2\Big(\alpha+\alpha x+x z+\beta+\beta x- (\beta^2+x+\alpha\beta+z\beta+ \alpha)\Big)& = \\
			2\Big(\alpha+\alpha x+x - x(\alpha + \beta)+\beta+\beta x- 
			(\beta^2+x+\alpha\beta+\beta - \beta(\alpha+\beta)+  \alpha)\Big)& = 0
		\end{align*}
		where in the second step we use $x \geq \beta$ since $\beta < \alpha$ as observed before this lemma. 
		
	\end{description}
This completes the proof of Lemma~\ref{le:high-jobs}.

\section{Conclusions}

In this paper we improve lower and upper bounds on the inefficiency ratio of stable equilibria in load balancing games, while still leaving a gap. We believe that our lower bound example can not be improved. In order to  prove this claim one has to consider global properties of potential optimizer, while in our proof only local properties have been used. 

\paragraph{Acknowledgments.} This work has been partially supported by the Swiss National Science Foundation (SNF) under the grant number
200021\_143323/1.

 \bibliographystyle{plainnat}
\bibliography{PotentialOptimum}

\end{document}